\documentclass[reqno,11pt]{amsart}
\usepackage{amsmath, latexsym, amsfonts, amssymb, amsthm, amscd}

\usepackage{xcolor,graphicx}

\usepackage{hyperref}

\usepackage{ucs}
\usepackage[applemac]{inputenc}
\fontfamily{Garamond} 
\advance\hoffset -.75cm                          

\usepackage[normalem]{ulem}

\numberwithin{equation}{section}

\newtheorem{Theorem}{Theorem}[section]

\newtheorem{Claim}[Theorem]{Claim}

\DeclareMathSymbol{\leqslant}{\mathalpha}{AMSa}{"36} 
\DeclareMathSymbol{\geqslant}{\mathalpha}{AMSa}{"3E} 
\DeclareMathSymbol{\eset}{\mathalpha}{AMSb}{"3F}     
\renewcommand{\leq}{\;\leqslant\;}                   
\renewcommand{\geq}{\;\geqslant\;}                   

\newcommand{\cF}{\ensuremath{\mathcal F}}
\newcommand{\cG}{\ensuremath{\mathcal G}}

\newcommand{\cR}{\ensuremath{\mathcal R}}
\newcommand{\cS}{\ensuremath{\mathcal S}}

\newcommand{\cV}{\ensuremath{\mathcal V}}

\newcommand{\bbD}{{\ensuremath{\mathbb D}} }
\newcommand{\bbE}{{\ensuremath{\mathbb E}} }

\newcommand{\bbP}{{\ensuremath{\mathbb P}} }

\newcommand{\bbR}{{\ensuremath{\mathbb R}} }

\newcommand{\gd}{\delta}
\newcommand{\gp}{\varphi}

\newcommand{\gs}{\sigma}

\newcommand{\Jstat}{{\bar J}}

\newcommand{\bJ}{{\bf J}}

\title{
A perturbative approach to the macroscopic fluctuation theory
}

\author[Bodineau]{Thierry Bodineau$^\dagger$}
\address[$\dagger$]{I.H.E.S., 35 Route de Chartres, 91440 Bures-sur-Yvette, France}

\author[Derrida]{Bernard Derrida$^\clubsuit$}
\address[$\clubsuit$]{Coll\`ege de France, 11 place Marcelin Berthelot, 75005 Paris, France}
\address[$\clubsuit$]{Laboratoire de Physique de l'Ecole Normale Sup\'erieure, ENS, Universit\'e PSL, CNRS, Sorbonne Universit\' e, Universit\' e de Paris, F-75005 Paris, France}

\begin{document}

\maketitle

\today

\begin{abstract}
In this paper, we study the stationary states of diffusive dynamics driven out of equilibrium by reservoirs. For a small forcing, the system remains close to equilibrium and the large deviation functional of the density  can be computed perturbatively by using the macroscopic fluctuation theory. This applies to general domains in $\mathbb{R}^d$ and diffusive dynamics with arbitrary transport coefficients. As a consequence, one can analyse the correlations at the first non trivial order in the forcing and show that, in general, all the long range correlation functions  are not equal to 0, in contrast to the exactly solvable models previously known. 
\end{abstract}

\section{Introduction}

The steady state  of  systems in contact with two reservoirs of particles at different densities
or two heat baths at different temperatures 
 is one  of the most basic situations one can consider  in the study of non-equilibrium systems.  
When the two reservoirs are at the same density or the two heat baths are at the same temperature, the steady state is nothing but the equilibrium  for which  a number of properties can be determined from the knowledge of the free energy.
For macroscopic systems in contact  with two or more reservoirs of particles at different densities,  one  quantity which generalizes the concept of free energy is  the large deviation functional ${\cS} $ of the density profiles. 
Over the last 25 years trying to determine this functional  has  motivated considerable efforts. 
We refer to \cite{BDGJL3,BDGJL4,BDGJL6, Derrida 2007} for some reviews on the topic.

\smallskip

For a one dimensional macroscopic system of length $L$  in contact, at its two ends,  with two reservoirs of particles at densities $\rho_\text{left}$ and $\rho_\text{right}$, this large deviation functional ${\cS   }(\{\gp(x)\})$  is defined by 
\begin{equation}
\text{Pro} [\{\gp(x)\}] \sim \exp \big( - L \,  {\cS   }(\{\gp(x)\}) \big) \quad \text{for large $L$},
\label{eq : LD  1d} 
\end{equation}
where  $\text{Pro}[ \{\gp(x)\} ]$ is the probability of observing a density profile $\gp(x)$ near position $L x$ with $0 < x < 1$. {One way of implementing this definition consists in slicing the macroscopic system  of size $L$ into   a large number $1/\delta $ of  large subsystems of size $L \delta $. Then $\text{Pro}(\{\gp(x)\})$ is simply the probability of finding $\gp(x)\times L \delta $ particles in the interval $\big(Lx, L (x + \delta )\big)$. In general  the large deviation functional $\cS$  depends both on the density profile $\gp(x)$ and on the densities $\rho_\text{left}$ and $\rho_\text{right}$ of the reservoirs.

The generalization to higher dimensions  can be done by considering a  lattice gas   described at the macroscopic level  by its density $\gp (x)$ at position $x$ in a bounded domain $\bbD \subset \bbR^d$ in dimension $d \geq 1$ and  assuming that  each  point $y \in \partial \bbD$
at the boundary  of the domain is connected to a reservoir at density $\rho_\text{b}(y)$. Then for a system of large size $L$ (one can think of $L$ as being the inverse of the lattice spacing), one expects the probability of observing a density profile $\gp = \{\gp(x)\}$ in the bulk  to be given   as in (\ref{eq : LD 1d})  by
\begin{equation}
\label{eq: LD d}
\text{Pro} [ \gp ] \sim \exp \big( - L^d {\cS   }( \gp ) \big)  .  
\end{equation}
As in the one dimensional case, the large deviation functional $\cS$ depends on the profile $\gp$ and on the densities $\{ \rho_\text{b}(y), \ y \in \partial \bbD \}$ of all the reservoirs.

\smallskip

To determine the large deviation functional $\cS$  in the non equilibrium case, two approaches have been followed so far.  One is based on the exact knowledge of the steady state measure, i.e. on the computation of  explicit expressions for the weights of all the microscopic configurations of the system.
This   has been initially implemented for the symmetric \cite{DLS1, DLS2, Tailleur, DHS,SS} and asymmetric exclusion processes \cite{DLS3, DLS4, ED}, then extended to other classes of microscopic models \cite{BE, CGT, CGG, CFGG,GKR}.
The other approach is at the heart of the macroscopic fluctuation theory initiated in \cite{BDGJL1, BDGJL2} and generalises the Freidlin-Wentzell theory \cite{Freidlin-Wentzell} by relating the steady state large deviation functional $\cS$ to the dynamical large deviations of the macroscopic density profiles.
This provides a robust framework to describe non-equilibrium large deviations for general diffusive dynamics beyond the exactly solvable models. In particular, it relies only on the knowledge of the macroscopic transport coefficients (the diffusion $D$ and the conductivity $\sigma$, see Appendix  \ref{Appendix : Transport coefficients}}) and not the detailed structure of the microscopic dynamics. 
Note that the singularities of the large deviation functional can   be  interpreted as phase transitions  \cite{BDGJL5, BK, BKP1, BKP3}.
The mathematical theory of pathwise large deviations \cite{KOV,KL}  has been extended to justify   rigorously this  macroscopic fluctuation theory \cite{BDGJL1 maths, BGL maths, BLM, BEL, BG,DmFG maths, Farfan, FLT, LV}, at least for some microscopic dynamics.

At the macroscopic scale, a general diffusive microscopic dynamics  is well approximated  by a non-linear diffusion equation with a small  noise.
Building on this coarse grained description, the macroscopic fluctuation theory identifies the large deviation functional $\cS ( \gp )$ defined in \eqref{eq: LD d} as the optimal cost of a dynamical fluctuation 
 to reach the density profile $\gp = \{ \gp(x) \}$ starting from the steady state profile  $\bar{\rho} = \{  \bar{\rho}(x) \}$. 
 This provides a (dynamical) variational formula to determine the stationary large deviation functional $\cS(\gp)$, known as the quasi-potential \cite{Freidlin-Wentzell}.
This strategy is recalled in Section \ref{sec:Optimal path to generate a deviation of the density}.
Therefore calculating $\cS$ boils down to analyzing non-linear equations (see (\ref{eq: evolution rho}),(\ref{eq: evolution H}))
to determine the  most likely path associated with the quasi-potential \eqref{eq: quasipotential}. 
Although these non-linear equations are in principle valid for arbitrary diffusive systems, one does not know in general how to solve them, except again for some particular diffusive systems  \cite{BGL}. 
It is striking to note that  the  
 only closed formula for $\cS$,   with long range correlations,  was obtained for one-dimensional processes with constant diffusion coefficient $D$ and a conductivity coefficient $\sigma$ which is a polynomial of order 2.
These coefficients fall into the class of exactly solvable models \cite{DLS2,CFGG} for which the functional can be computed explicitly.  Out of equilibrium,  $\cS$ is also explicit for dynamics related to the Zero-Range process \cite{De masi Ferrari}, but in this model, there is no  long range correlations  out of equilibrium.

\smallskip

The macroscopic fluctuation theory is a generalisation to non-equilibrium systems of the Onsager-Machlup theory \cite{Onsager}.
At equilibrium, i.e. when  all the reservoirs in contact with the system are at the same density $\rho_\text{b}(y)=\rho_*$ for all $y \in \partial \bbD$, the optimal path (we will call it the excitation path)
followed by the system to reach any specified density profile $\gp$ starting from the steady state profile 
$\bar{\rho}$ is exactly the same as the relaxation path followed by the system when it relaxes from $\gp$ to $\bar{\rho}$. From this reversibility property, the dynamical cost of the optimal path can be computed and the Onsager-Machlup theory 
yields a dynamical interpretation of the large deviation functional for  equilibrium systems.
For systems with short range interactions and in absence of phase transitions, (the only cases we are discussing here),  
the large deviation functional $\cS^{eq} $  of the density is known: it is related  to the free energy prescribed by Gibbs theory (see \cite[Section 4, Eq. (18)]{Derrida 2007} and also \cite{Lanford, Ellis})  
\begin{equation}
\label{eq: LD equilibrium}
\cS^{eq} (\gp) = \int_\bbD dx \; 
S_{ \rho_*} \big( \gp (x) \big)
\quad \text{with} \quad 
S_{ \rho_*} (\rho)=
f ( \rho) - f ( \rho_* ) - (\rho -\rho_*) f' ( \rho_* )  ,
\end{equation}
where $f(\rho)$ is the free energy   per unit volume of the lattice gas at density $\rho$, i.e. 
\begin{equation}
\label{def: free energy}
f(\rho ) = -\lim_{V\to\infty} {\log Z_{V}(V \rho) \over V}
\end{equation}
with $Z_V(N)$ is the partition function of a system of $N$ particles in a volume $V$. It is easy to check  that  $\cS$ vanishes for  the constant equilibrium profile $\gp(x)=\rho_*$ and that $\cS >0$ for all other profiles (this is a consequence of  the fact that $f$ is strictly convex as we exclude the case of phase transition).

\medskip

From the analogy between the free energy and the large deviation functional for equilibrium systems, 
it is natural to consider the non-equilibrium large deviation functional $\cS$ as a substitute for the free energy defined   by the Gibbs theory.
Even though, there is (so far) no explicit solution for the large deviation functional, the dynamical variational principle can be studied in some perturbative regimes. In particular for density profiles $\gp$ close to the stationary profile $\bar \rho$, the excitation path can be analysed by perturbation theory and the correlations for the stationary measure can be recovered by expanding $\cS$ close to $\bar \rho$ \cite{BDGJL2}. 
A different perturbative approach has been implemented in \cite{VWijland Racz} in the case of  weakly interacting systems.

In this paper, we consider another perturbative regime and study the large deviation functional $\cS$  associated with  a small variation of the reservoir  densities, i.e. $|\rho_\text{b}(y)-\rho_\text{b}(z)| \ll 1$ for all $y$ and $z \in \partial \bbD$.
In this case, the system should remain close to equilibrium and one expects that, in agreement with Onsager-Machlup theory,  the excitation and the relaxation paths are close to each other. Based on this idea, the goal of the present work is to develop a perturbative method to calculate the large deviation functional  $\cS$ by using the macroscopic fluctuation theory.
The main result of the present paper is an expression of the large deviation functional (see Claim \ref{Claim: HJ})
for the leading corrections to (\ref{eq: LD equilibrium})  for systems near equilibrium. 
As expected for generic non-equilibrium systems, this functional is in general  non local due to the occurence of long range correlations \cite{DKS,Spohn}.

\section{Optimal path to generate a deviation of the density}
\label{sec:Optimal path to generate a deviation of the density}

In this section, we recall the main features of the macroscopic fluctuation theory \cite{BDGJL1, BDGJL2} and in particular the Euler-Lagrange equations 
(\ref{eq: evolution rho}, \ref{eq: evolution H}, \ref{eq: cost variational principle})
to be solved to determine the quasi-potential, namely the non-equilibrium large deviation functional $\cS$ introduced in \eqref{eq: LD d}.

\medskip

For a large diffusive system  ($L \gg1$),
 the macroscopic density $\rho (t,x)$ is linked to the particle current $q (t,x) \in \bbR^d$ by the conservation law
 \begin{equation}
\label{eq: conservation law}
x \in \bbD, \qquad 
\partial_t \rho(t,x) = - \nabla \cdot  q (t,x)  ,
\end{equation}
 which itself  evolves according to the following fluctuating hydrodynamic equation
\begin{equation}
\label{eq: fluctuating hydro}
 q (t,x) = - D \big( \rho (t,x) \big) \nabla \rho (t,x) + \frac{1}{L^{d/2}} \sqrt{ \sigma \big( \rho (t,x) \big) } \eta(t,x) ,
\end{equation}
where $\eta \in \bbR^d$ is a space-time white noise  (see \cite[Chapter 2.9, Eq. (2.117)]{Spohn book}).
Furthermore at the boundary $\partial \bbD$ the density is fixed  by the densities of the reservoirs $\rho(t,y)= \rho_\text{b}(y)$ for $y \in \partial \bbD$.
The dynamics 
(\ref{eq: fluctuating hydro})
is characterised by the diffusion coefficient $D(\rho)$ and the conductivity $\sigma(\rho)$ which are functions depending only on the density $\rho$. We assume that the dynamics 
is isotropic so that $D, \sigma$ are scalar.
We refer to   Appendix \ref{Appendix : Transport coefficients} for a discussion on  the transport coefficients 
$D, \sigma$ and for a derivation of the  Einstein relation : for any density $\rho$, the two coefficients are related by 
\begin{equation}
\label{eq : Einstein relation}
f''(\rho) \, \sigma(\rho)= 2 D(\rho),
\end{equation}
where $f$ is the equilibrium free energy introduced in \eqref{def: free energy}.
In the literature  \cite[Eq. (27)]{Derrida 2007}, $1/f''(\rho)$ is usually called the compressibility.
The stationary profile $\bar{\rho}(x) $ is solution of
\begin{equation}
\label{eq: rho bar}
x \in \bbD, \qquad 
\nabla. \Big(D\big(\bar{\rho}(x)\big)  \, \nabla \bar{\rho}(x) \Big)=0  \quad \quad   \quad \text{ with}  \quad  \bar{\rho}(y)= \rho_\text{b}(y) \quad  \text{for} \quad   y \in \partial \bbD.
\end{equation}
 In this paper, we consider systems such that $D(\rho) \geq c >0$ is uniformly bounded from below.

\medskip

Using the fluctuation of the noise in \eqref{eq: fluctuating hydro},  the probability of observing an atypical density and a current trajectory  $( \rho(t), q(t) )= \{ \rho(t,x), q(t,x) \}_{x \in \bbD}$ during any macroscopic time interval $[\tau_1,\tau_2]$ and starting at $\tau_1$ from the initial profile $\rho(\tau_1)$  is given by
\begin{equation}
\label{eq: LD hydro}
\bbP_{[\tau_1,\tau_2]} \Big( \text{observing $(\rho(t), q(t))$ for $\tau_1 \leq t \leq \tau_2$} \; \big| \; \rho(\tau_1) \Big) 
\sim
\exp \left( - L^d \;  \cF_{[\tau_1,\tau_2]} (\rho,q) \right) ,
\end{equation}
under the constraint  of the conservation law \eqref{eq: conservation law}  and with 
\begin{equation}
\label{eq: LD hydro1}
\cF_{[\tau_1,\tau_2]} (\rho,q) =  \int_{\tau_1}^{\tau_2} dt \int_\bbD dx 
\frac{\big(q (t,x) + D \big( \rho (t,x) \big) \nabla \rho (t,x) \big)^2}{2 \sigma \big( \rho (t,x) \big)} .
\end{equation}

\begin{figure}[htbp]
\begin{center}
\includegraphics[width=4in]{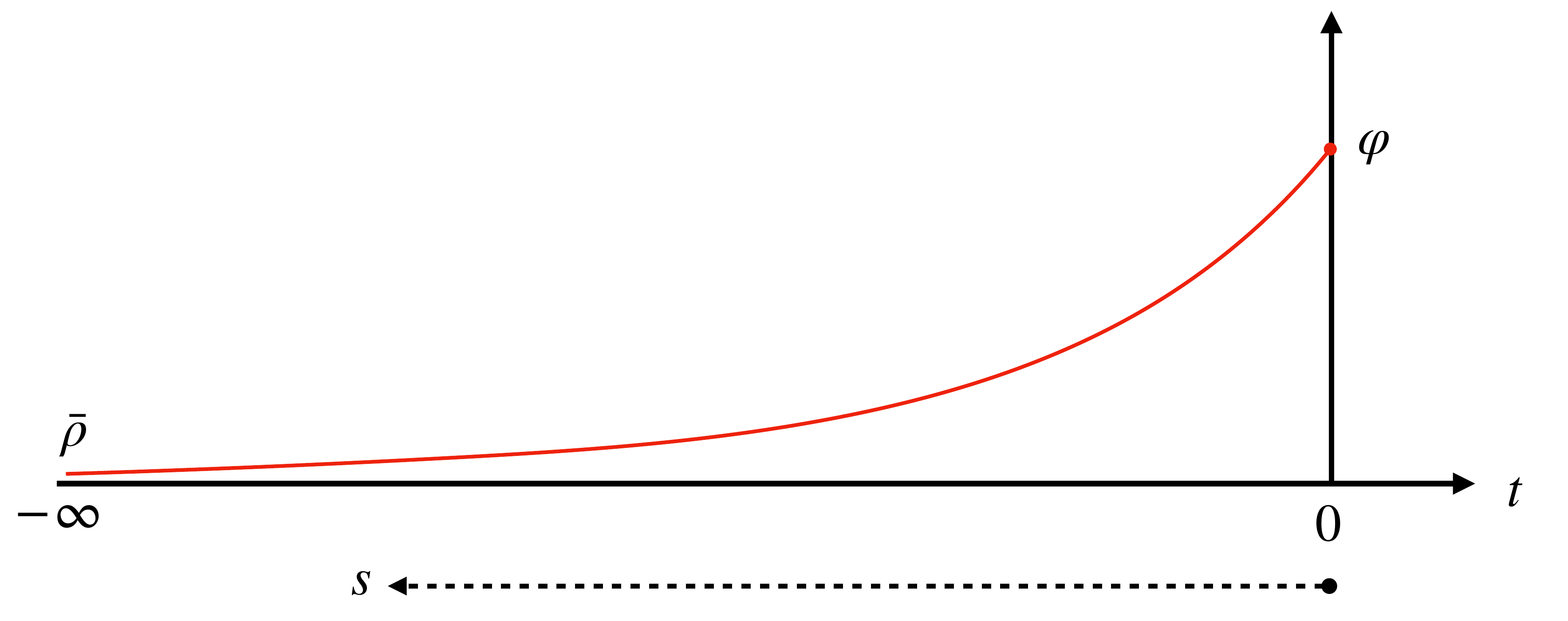} 
\caption{\small On the figure, an optimal trajectory $( \rho(t) )_{t \leq 0}$  solution of 
 \eqref{eq: evolution rho} is depicted. It interpolates between $\bar \rho$ at time $-\infty$ and $\varphi$ at time 0. 
 This can be viewed as  a relaxation path of the density profile from $\gp$ to $\bar \rho$ by reversing time $s = -t$ (see the dashed arrow on the figure and \eqref{eq: reverse optimal trajectory}).
}
\label{fig: optimal path}
\end{center}
\end{figure}
The large deviation functional $\cS$ is given by the quasi-potential  (see \cite{BDGJL2, Freidlin-Wentzell}), i.e. the smallest  cost of a dynamical fluctuation :
\begin{equation}
\label{eq: quasipotential}
\cS (\gp) 
= \inf \big\{ \cF_{[-\infty,0]} (\rho,q); \quad \rho(0,x) = \gp (x) \big\},
\end{equation}
where the time interval is now infinite $]-\infty,0]$ and the infimum is taken over all the trajectories with density $\gp(x)$ at the final time $t =0$ and such that  $\rho(t) \to \bar{\rho}$ as $ t \to- \infty$. 

 According to the macroscopic fluctuation theory  \cite{BDGJL2} (and see 
 Appendix \ref{Appendix : Equations for the optimal profiles} for a derivation), 
 the current which minimizes (\ref{eq: quasipotential})
can be written as  
\begin{equation}
\label{eq: optimal current}
q(t,x)=- D \big( \rho \big) \nabla \rho + \sigma\big( \rho \big) \nabla H
\end{equation}
with $H(t,y) = 0$ for $y \in \partial \bbD$.
This leads to (see (\ref{eq: LD hydro1}),(\ref{eq: quasipotential}))
\begin{equation}
\label{eq: cost variational principle}
\cS (\gp) 
= \frac{1}{2} \int_{-\infty}^0 dt 
\int_\bbD dx \,  \sigma \big( \rho (t,x)\big) \, \big( \nabla H (t,x)\big)^2 .
\end{equation}
 From (\ref{eq: optimal current}) and the conservation law \eqref{eq: conservation law},
the density evolves as 
\begin{equation}
\label{eq: evolution rho}    
\partial_t \rho  = \nabla \cdot ( D(\rho) \nabla  \rho) - \nabla \cdot  ( \sigma(\rho) \nabla  H)
\quad \text{ with} \quad  \rho(t,y)= \rho_\text{b}(y) \quad  \text{for}    \quad     y \in \partial \bbD.
\end{equation}
As $\sigma (\rho)>0$,  note that
 once $\rho(t)$ is given,
  (\ref{eq: evolution rho}) fully determines
the function $H(t)$ 
under the constraint that  $H(t,y)=0$ for $y \in \partial \bbD$.
In Appendix \ref{Appendix : Equations for the optimal profiles}, it is also recalled that the evolution of $H$ for a path which minimizes the cost (\ref{eq: cost variational principle}) is given by
\begin{equation}
\label{eq: evolution H}
\partial_t H = -  D(\rho) \Delta  H - \frac{1}{2} \sigma'(\rho) ( \nabla  H)^2 .
\end{equation}
The optimal evolution equations \eqref{eq: evolution rho}    and \eqref{eq: evolution H} are standard in the macroscopic fluctuation theory \cite{BDGJL2}. 
Note that this system is constrained by the conditions on the density at time $0$ and $-\infty$. At time $t=0$, one also has (see \eqref{eq: cF et H appendix})
\begin{equation}
\label{eq: optimal H 0}
x \in \bbD, \qquad 
 H (0,x)= \frac{\delta \cS}{\delta \gp (x)} (\gp) ,
\end{equation} 
where the derivative is wrt the final density $\gp$. By the optimality of the path \eqref{eq: evolution rho}, this relation remains true at any time  $t \le 0$ so that  in \eqref{eq: evolution rho}
\begin{equation}
\label{eq: optimal H}
x \in \bbD, \qquad H (t,x)= \frac{\delta \cS}{\delta \gp (x)} (\rho (t)). 
\end{equation}

\medskip

We conclude this section by recovering the optimal excitation path at equilibrium from the previous framework.
The equilibrium large deviation functional is explicit \eqref{eq: LD equilibrium} so that \eqref{eq: optimal H} implies the following relation along the optimal path \eqref{eq: evolution rho}
\begin{equation}
\label{eq: optimal H 2}
H (t)= \frac{\delta \cS^{eq} }{\delta \gp} (\rho(t))
=  f'(\rho(t))-f'(\rho_*).
\end{equation} 
Plugging \eqref{eq: optimal H} in the optimal density fluctuation \eqref{eq: evolution rho} and 
using Einstein relation \eqref{eq : Einstein relation}, we see that the optimal path  is just the time reversal of the relaxation trajectory
\begin{equation}
\label{eq: conservation law Onsager-Machlup}
\partial_t \rho  = - \nabla \cdot ( D(\rho) \nabla  \rho) 
 \quad \text{for}  \quad  t \leq 0 ,
\end{equation}
in agreement with Onsager-Machlup theory \cite{Onsager}.

\section{Perturbative  non-equilibrium  large deviations}
\label{sec: Perturbative  non-equilibrium  large deviations}

\subsection{Statement of the large deviation functional}

We consider now small variations of the density $\rho_\text{b}$ at the reservoirs.
In this case, the stationary density $\bar \rho$ defined in \eqref{eq: rho bar} is weakly  space dependent 
so that the stationary current
\begin{equation}
\label{eq: definition current}
x \in \bbD, \quad 
\Jstat(x) = - D \big( \bar \rho (x) \big) \nabla \bar \rho (x)
\end{equation}
 is assumed to be small 
\begin{equation}
\label{eq: assumption current}
\bJ = \sup_{x \in \bbD}  | \Jstat(x)|  \ll 1.
\end{equation}
We are now going to use the  dynamical approach described in Section \ref{sec:Optimal path to generate a deviation of the density} 
to compute perturbatively in $\bJ \ll 1$ the large deviation functional $\cS$ defined in \eqref{eq: quasipotential}.

\medskip

As the system remains close to equilibrium, the dominant part of the large deviation functional will be 
a perturbation of $\cS^{eq}$ defined in  \eqref{eq: LD equilibrium}. Thus we set 
\begin{equation}
\label{eq: local equilibrium  d 1}
\cS^{local} (\gp) = \int_\bbD dx S_{\bar \rho(x)}(\gp(x))
\quad \text{with} \quad 
S_{\bar \rho}(\gp) = 
f(\gp)-f(\bar\rho)- (\gp - \bar\rho) f'(\bar\rho),
\end{equation}
which can be rephrased in terms of the transport coefficients by the Einstein relation 
\eqref{eq : Einstein relation}
\begin{equation}
\label{eq: local equilibrium  d}
S_{\bar \rho}'(\gp) = \int_{\bar \rho}^\gp  dr \frac{2 D(r)}{\sigma(r)}.
\end{equation}

The main result of this paper is an expansion of the non-equilibrium large deviation functional $\cS$ at order
$\bJ^3$. This encodes the long range correlations in the steady state. For this, we consider the function 
$\Gamma : \bbR^+ \mapsto \bbR$ such that  $\Gamma'(\rho) = D(\rho)$ and we define
the functional 
\begin{align}
\label{eq: definition V bis d}
\cV(\gp) & = - 2  \int_0^\infty ds \;
\int_\bbD dx  \frac{  ( \Jstat  )^2  }{\sigma(\bar \rho)^2} \left( \sigma(\rho (s) ) - \sigma(\bar \rho)   -  \frac{\Gamma(\rho(s))- \Gamma(\bar \rho)}{ D(\bar \rho)}  \sigma'(\bar \rho) \right),
\end{align}
where $(\rho (s))_{s \geq 0}$ is the solution of 
\begin{equation}
\label{eq: reverse optimal trajectory}
s \geq 0, x \in \bbD, \qquad 
\partial_s \rho (s,x)  =  \nabla \cdot ( D(\rho(s,x) ) \nabla  \rho(s,x) ) + 2 \Jstat (x) \cdot \nabla \left(  \frac{\sigma(\rho(s,x) ) }{\sigma(\bar \rho(x))} \right)
\end{equation}
starting from $\rho(0,x) = \gp(x)$ (see Figure \ref{fig: optimal path}).
Note from \eqref{eq: rho bar} that  the stationary density $\bar \rho$ is a time independent solution of \eqref{eq: reverse optimal trajectory} so that  $\cV(\bar \rho) = 0$.

\begin{Claim}
\label{Claim: HJ}
The large deviation cost $\cS(\gp)$ for observing a density profile $\gp$ in the stationary  state is of the form 
\begin{equation}
\label{eq: LD ordre 3 d}
\cS (\gp) = \cS^{local} (\gp)  + \cV(\gp) + O (\bJ^4),
\end{equation}
where $\cV$ is defined in \eqref{eq: definition V bis d}.
\end{Claim}
This claim is derived in  Section \ref{sec: Optimal trajectories}.

\medskip

Note that for the Zero-Range process (see \cite{KL}),  the transport coefficients satisfy the relation $ \sigma'(\rho) =2 D(\rho)$ so that  $\Gamma (\rho) = \sigma (\rho)/2$. For this choice,  the non-local contribution $\cV$  \eqref{eq: definition V bis d} to the large deviation  $\cS$ vanishes in agreement with the absence of long range correlations in this model \cite{De masi Ferrari}.

\subsection{Derivation of Claim \ref{Claim: HJ}}
\label{sec: Optimal trajectories}

We are now going to check Claim \ref{Claim: HJ} by computing 
 the optimal trajectories in order to provide a solution of the variational problem \eqref{eq: quasipotential}.

We recall the optimal equations \eqref{eq: evolution rho}    and \eqref{eq: evolution H} associated with 
the variational problem \eqref{eq: quasipotential} on the time inteval $]-\infty , 0]$ :
\begin{equation}
\label{eq: evolution equations d}
t \leq 0, \qquad 
\begin{cases}
\partial_t \rho  = \nabla \cdot ( D(\rho) \nabla  \rho) - \nabla  \cdot( \sigma(\rho) \nabla  H) , \\
\partial_t H  = -  D(\rho) \Delta  H - \frac{1}{2} \sigma'(\rho) ( \nabla  H)^2,
\end{cases}
\end{equation}
with  fixed density  $\rho(t) = \rho_\text{b}$ on $\partial \bbD$ imposed by the reservoirs  and 
Dirichlet boundary conditions $H(t) =0$.
The main difficulty is that only the density is prescribed initially (for time $t \to -\infty$) and at the final time ($t=0$):
\begin{equation}
\label{eq: boundary conditions optimal}
x \in \bbD, \qquad \rho(t,x) \xrightarrow[t \to -\infty]{} \bar \rho(x), \qquad \rho(0,x) = \gp(x).
\end{equation}
The control $H$ at time $t=0$ has be determined and from this, the functional
$\cS$ will then  be obtained thanks to the relation \eqref{eq: cF et H appendix}
\begin{equation}
\label{eq: H0 et S}
\frac{\delta \cS}{\delta \gp} (\gp) = H(t = 0).
\end{equation}
 To solve \eqref{eq: evolution equations d}, we use that at equilibrium $H$ is explicit 
 (see (\ref{eq: optimal H 2},\ref{eq: local equilibrium  d}))
and we look for an ansatz of the following form to take into account the new stationary profile $\bar \rho (x)$
\begin{align}
\label{eq: ansatz H d}
x \in \bbD, t \leq 0, \quad
  H (t,x) 
=   \int_{\bar \rho(x)}^{ \rho (t,x)} du \frac{2 D(u)}{\sigma(u)} 
 +  h(t,x),
\end{align}
where $h$ is a correction vanishing with $\bJ$ and $h(t,x) = 0$ for $x \in \partial \bbD$.

\begin{Claim}
\label{Claim: perturbative equations}
The optimal evolution equations \eqref{eq: evolution equations d} can be rewritten in terms of $h$ as 
\begin{equation}
\label{eq: perturbative evolution equations 2}
\begin{cases}
\partial_t \rho  = - \nabla\cdot ( D(\rho) \nabla  \rho) - 2  \Jstat  \cdot  \nabla \left(  \frac{\sigma(\rho) }{\sigma(\bar \rho)} \right) - \nabla  \cdot ( \sigma(\rho) \nabla  h), \\
\partial_t h  =    D(\rho) \Delta h 
 - 2 \frac{\sigma'(\rho)}{\sigma(\bar \rho)}  \Jstat \cdot  \nabla h 
-   \frac{2 \big(\Jstat \big)^2  }{\sigma(\bar \rho)^2}    \left(  \sigma'(\rho)   
 -  \frac{D(\rho)}{ D(\bar \rho)}  \sigma'(\bar \rho)  \right)  
 -\frac{1}{2} \sigma'(\rho) \left(   \nabla h \right)^2,
\end{cases}
\end{equation}
with boundary conditions on $\partial \bbD$ given by $\rho_\text{b}$ for the density and $0$ for $h$.
We stress that in the equation above the stationary density $\bar \rho$ and current $\Jstat$ depend on $x \in \bbD$.
\end{Claim}

We state now some consequences of Claim \ref{Claim: perturbative equations} and refer to
 Appendix \ref{Appendix: Check of Claim} for its derivation.
Equations \eqref{eq: perturbative evolution equations 2} encodes the optimal path  for $t \in ]- \infty,0]$.
It is convenient to assume that the functions $\rho(0,x) = \gp(x)$, $h(0,x)$ are prescribed at time 0
and to consider instead the time reversed equation with $s \in [0, +\infty[$ 
\begin{equation}
\label{eq: perturbative evolution equations reversed}
\begin{cases}
\partial_s \rho  =  \nabla ( D(\rho) \nabla  \rho)  + 2  \Jstat  \cdot  \nabla \left(  \frac{\sigma(\rho) }{\sigma(\bar \rho)} \right) + \nabla ( \sigma(\rho) \nabla  h) , \\
\partial_s h  = -   D(\rho) \Delta h 
 + 2 \frac{\sigma'(\rho)}{\sigma(\bar \rho)}  \Jstat \cdot  \nabla h 
+   \frac{2 \big(\Jstat \big)^2  }{\sigma(\bar \rho)^2}    \left(  \sigma'(\rho)   
 -  \frac{D(\rho)}{ D(\bar \rho)}  \sigma'(\bar \rho)  \right)  
 +\frac{1}{2} \sigma'(\rho) \left(   \nabla h \right)^2.
\end{cases}
\end{equation}
In the following $s$ stands for the time reversed parameter (see Figure \ref{fig: optimal path}).

If $h$ was of order $\bJ$, then its evolution at order $\bJ$ would be 
$\partial_s h  = - D(\rho) \Delta h$ which is unstable. Thus $h(0)=0$ at order $\bJ$ and we will look for higher order corrections.  Rewriting the evolution  \eqref{eq: perturbative evolution equations reversed} of $h$ up to order $\bJ^3$ gives :
\begin{equation}
\label{eq: perturbative equations h order 3}
\partial_s h  = -   D(\rho) \Delta h 
 + 2 \frac{\sigma'(\rho)}{\sigma(\bar \rho)}  \Jstat \cdot  \nabla h 
+   \frac{2 \big(\Jstat \big)^2  }{\sigma(\bar \rho)^2}    \left(  \sigma'(\rho)   
 -  \frac{D(\rho)}{ D(\bar \rho)}  \sigma'(\bar \rho)  \right)  ,
\end{equation}
where $\rho$ satisfies the autonomous equation \eqref{eq: reverse optimal trajectory} recalled below
\begin{equation}
\label{eq: perturbative  equations rho}
\forall s \geq 0, \qquad 
\partial_s \rho  =  \nabla ( D(\rho) \nabla  \rho) +  2  \Jstat  \cdot  \nabla \left(  \frac{\sigma(\rho) }{\sigma(\bar \rho)} \right) 
\end{equation}
with $\rho(0) = \gp$. Indeed, it is enough to keep the corrections of $\rho$ at order $\bJ$ 
to take into account all the terms up to order  $\bJ^3$ in \eqref{eq: perturbative equations h order 3}.

Due to the negative sign in front of the diffusion term in \eqref{eq: perturbative equations h order 3},
the function $h$ obeys an unstable equation with a source term (of order $\bJ^2$) determined by $(\rho(s))_{s \geq 0}$, the solution of  \eqref{eq: perturbative  equations rho}. We are going to show that this instability fully determines the initial data $h(0)$.
Given the solution $(\rho(s))_{s \geq 0}$  of  \eqref{eq: perturbative  equations rho}, 
let us introduce for $s'\geq s$ the Green function $G(s',y|s,x)$ solution of
\begin{equation}
\label{eq: def Green function}
{d G \over ds'} = \Delta_y \Big(D(\rho(s',y)) G  \Big)
+ 2 \bar{J}(y) \cdot \nabla_y \left({\sigma'( \rho(s',y)) \over \sigma(\bar\rho(y)) }G \right)  \ \ \ \ \ \text{with} \ \ \ G(s,x| s,y)=\delta(x-y).
\end{equation}
This semigroup satisfies also
$$
{d G \over ds} = - D(\rho(s,x)) \Delta_x G 
+ 2 \bar{J}(x) \cdot  \left({\sigma'( \rho(s,x)) \over \sigma(\bar\rho(x)) } \right)  \nabla_x G 
\ \ \ \ \ \text{with} \ \ \ G(s,y| s,x)=\delta(x-y).
$$
Thus the solution of \eqref{eq: perturbative equations h order 3} is given by 
\begin{equation}
\label{eq: h(s)}
h (s,x)  =   - 2  \int_s^\infty  ds' \, \int_\bbD dy \, G(s',y|s,x) \left(  \frac{ \big(\Jstat(y) \big)^2  }{\sigma(\bar \rho (y))^2}    \left(  \sigma'(\rho (s',y) )   -   \frac{D(\rho(s',y) )}{ D(\bar \rho(y))}  \sigma'(\bar \rho(y))  \right) \right) .
\end{equation}
In particular, this determines the value of $h(0,x)$.

\medskip

Recalling identity \eqref{eq: H0 et S} and the definition \eqref{eq: ansatz H d}
of $H$, it is enough to identify the long range part $\cV$  of the large deviation functional $\cS$.
We are going to show that $h(0,x)= \frac{\partial \cV}{\partial \gp(x)} (\gp)$ where the functional $ \cV$ \eqref{eq: definition V bis d} is recalled below
\begin{align}
\label{eq: definition V ter d}
\cV(\gp) & = - 2  \int_0^\infty ds \;
\int_\bbD dx  \frac{  ( \Jstat (x) )^2  }{\sigma(\bar \rho  (x))^2} \left( \sigma(\rho(s,x)) - \sigma(\bar \rho (x))   -  \frac{\Gamma(\rho(s,x))- \Gamma(\bar \rho (x))}{ D(\bar \rho (x))}  \sigma'(\bar \rho  (x)) \right),
\end{align}
with $( \rho (s))_{s \geq 0}$ the solution of \eqref{eq: perturbative  equations rho}.
For a  perturbation of the initial data from $\gp$ to $\gp + \psi$, the solution of \eqref{eq: perturbative  equations rho} has a correction evolving according to the linearised equation 
\begin{align}
\label{eq: linearised density reversal}
\partial_s \phi  =  \Delta ( D(\rho) \phi  ) 
+  2  \Jstat  \cdot  \nabla \left(  \frac{\sigma'(\rho)  }{\sigma(\bar \rho)} \phi \right), 
\qquad \phi(0) = \psi.
\end{align}
Thus from \eqref{eq: def Green function}, we deduce that the  perturbation of the density is determined by
\begin{align}
\label{eq: linearised density reversal 2}
s \geq 0, \qquad  \phi(s,y) = \int_\bbD dx \,  G(s,y|0,x)  \psi (x) .
\end{align}
Plugging this  in $\cV(\gp)$ \eqref{eq: definition V ter d} leads to 
\begin{align*}
\lim_{\gp \to 0} \frac{1}{\gd} & \big( \cV(\gp + \delta \,  \psi ) - \cV(\gp) \big)
 = - 2  \int_0^\infty ds \;
\int_\bbD dy  \frac{  ( \Jstat (y)  )^2  }{\sigma(\bar \rho(y))^2} \left( \sigma' (\rho(s,y))     -  \frac{D(\rho(s,y))}{ D(\bar \rho(y))}  \sigma'(\bar \rho(y)) \right) \phi (s,y)\\
&= - 2  \int_\bbD dx \;  \psi  (x) \int_0^\infty ds \;
\int_\bbD dy \, G(s,y|0,x) \,  \frac{  ( \Jstat (y)  )^2  }{\sigma(\bar \rho(y))^2} \left( \sigma' (\rho(s,y))     -  \frac{D(\rho(s,y))}{ D(\bar \rho(y))}  \sigma'(\bar \rho(y)) \right)\\
& = \int_\bbD dx  \psi (x) \, h(0,x),
\end{align*}
where we used \eqref{eq: h(s)} to identify $h$ in the last equality. 
This implies $h(0,x)= \frac{\delta \cV}{\delta \gp (x)} (\gp)$ in agreement with \eqref{eq: H0 et S}.
This completes the identification of the large deviation functional $\cS$ at order $\bJ^3$.

\section{Cumulants}

From (\ref{eq: local equilibrium  d 1}), (\ref{eq: local equilibrium  d}), 
(\ref{eq: definition V bis d}) and (\ref{eq: LD ordre 3 d}), one can determine at order $\bJ^3$ all the correlations of the densities in the non-equilibrium steady state
for arbitrary
$D(\rho)$ and $\sigma(\rho)$. In this section, to avoid too complicated expressions, we will restrict the discussion to the order $\bJ^2$ and to the particular case where 
\begin{equation}
D(\rho)=1.
\label{D=1}
\end{equation}
Taylor expanding $f$ and $\sigma$ in (\ref{eq: LD ordre 3 d}), we get that at order $\bJ^2$
\begin{align}
\label{eq : S at order J2}
\cS (\gp) = 
 \sum_{k \ge 2} {1 \over k!} 
\int_\bbD dx   
\left[   
  f^{(k)}(\bar\rho(x)) \,   (\gp(x)-\bar{\rho}(x))^k 
   -2  \sigma^{(k)}(\bar{\rho} )  
\, {  \bar{J}^2(x)  \over \sigma(\bar{\rho} )^2} \, 
 I_k(\{\gp(x)\},x)  
\right] ,
\end{align}
where  the coefficients are given by 
\begin{equation}
I_k (\{\gp(x)\},x)= \int_0^\infty  ds \big( \rho(s,x)-\bar{\rho}(x) \big)^k
\quad \text{with} \quad s\geq 0, x \in \bbD, \ \  
\partial_s \rho(s,x) = \Delta \rho(s,x), \ \rho(0) = \gp.
\label{Ik}
\end{equation}
 Above $\rho(s)$ follows equation \eqref{eq: reverse optimal trajectory} at leading order  (indeed the drift at order $\bJ$ contributes only to $\cS$ at order $\bJ^3$ so it can be forgotten).  
It is striking to note that   \eqref{eq : S at order J2} simplifies greatly for $\sigma$ quadratic : in this case 
the long range part of the functional  $\cV$ is quadratic  in $\gp - \bar \rho$ at order $\bJ^2$ (see 
 \eqref{eq : S at order J2} and \eqref{Ik}). 
The only explicit form of the large deviation functional $\cS$ with a genuine long range interaction
 known so far are  for one-dimensional models with this choice of transport coefficients \cite{DLS2,BGL}, like the SSEP or the KMP model.
A notable feature of these models is that only the pair correlations are non-zero at order $\bJ^2$ \cite{DLS5}.
As we shall see, this property is related to the quadratic structure of $\cV$ when $D=1$ and $\sigma$ is quadratic.
For general  $\sigma$, 
the series \eqref{eq : S at order J2}
 may be infinite and all the long range correlation functions are of order $\bJ^2$.

\medskip

In \eqref{eq : S at order J2} all the non-local contributions to $\cS$ come from the integrals $I_k$ which we are going to analyse now.
The density profile in \eqref{Ik} has the form
$$
s\geq 0, x \in \bbD,
\qquad
\rho(s,x) =  \bar \rho(x) + \int_\bbD dy \, G(s,x|0,y) \, \big(\gp(y)-\bar\rho(y)\big),
$$
where the Green function (already introduced in \eqref{eq: def Green function}) is given in 
 the present  case as the solution of
$$
{d G(s,x|\tau,y) \over ds} = \Delta \,  G(s,x| \tau,y) \quad \text{with} \quad G(\tau,x|\tau,y)= \delta (x-y)
$$
and $G(s,x|\tau,y)$ vanishes for $x \in \partial \bbD$.
One then gets  the following expressions of the integrals $I_k$ 
\begin{align}
I_2(x)  & =  \int_\bbD dx_1 r(x_1) \int_\bbD dx_2 
 r(x_2)  \int_0^\infty  ds \, G(s,x|0,x_1) \, G(s,x|0,x_2)
 \label{eq : I2}
\end{align}
and more generally
\begin{equation}
\label{eq : Ik}
I_k(x) =  \int_\bbD dx_1 r(x_1) \cdots  \int_\bbD dx_k 
 r(x_k) \, \int_0^\infty  ds \, G(s,x|0,x_1) \cdots \, G(s,x|0,x_k) ,
\end{equation}
where $$r(x)= \gp(x) - \bar\rho(x).$$
Therefore expanding $\cV$  in powers, we get
\begin{equation}
\cV (\gp) = \sum_{k \ge 2}{1 \over k!}  \int_\bbD dx_1 \, r(x_1)\cdots \int_\bbD d x_k\,r(x_k) 
\cV_k(x_1,\cdots x_k) ,
\label{eq : expansion of V 0}
\end{equation}
with leading order in $\bJ^2$
\begin{align}
 \cV_k(x_1,\cdots x_k)= 
 -2 \int_\bbD dx \, \sigma^{(k)}(\bar{\rho})
\, {  \bar{J}^2 (x)  \over \sigma(\bar{\rho})^2} \,
\int_0^\infty ds \, G(s,x| 0,x_1) \cdots  G(s,x|0,x_k).
\label{eq : def cVk}
\end{align}

We stress that the  series expansion \eqref{eq : expansion of V 0}
 is not perturbative in $r = \gp - \bar\rho$ 
provided the derivatives $\sigma^{(k)}(\bar{\rho})$ decay sufficiently fast so that the series converges. We will use now this expansion to recover the cumulants of the stationary measure.
The link between the non-locality of $\cS$ and the long range correlations is discussed in Appendix 
\ref{Appendix : Non-locality of S and long range correlations}.
From (\ref{eq : 2pts}), (\ref{eq : 3pts}) and (\ref{eq : kpts}), one can determine the long range part
of the truncated correlations  when $x_1,\cdots x_k$ are all different
\begin{equation}
\label{eq : kpts-a}
 \langle \gp(x_1) \gp(x_2 ) \cdots  \gp(x_k) \rangle_c = - {
 \sigma(\bar\rho(x_1)) \cdots \sigma(\bar\rho(x_k))
 \over 2^k \,  L^{(k-1)d}}  
\ \cV_k(x_1 \dots x_k),
\end{equation}
 where we have used that  $f''(\rho) = 2/ \sigma(\rho)$ by (\ref{eq : Einstein relation}) and the fact that $D(\rho)=1$. Finally, setting 
\begin{equation}
\label{eq : G kpts-b}
 \mathbf{G} (x_1 , \cdots  x_k) 
 =
 \int_\bbD dx
\,   \bar{J}^2 (x) 
\int_0^\infty ds
G(s,x| 0,x_1) \cdots  G(s,x|0,x_k),
\end{equation}
one gets
\begin{equation}
\label{eq : kpts-b}
 \langle \gp(x_1) \gp(x_2 ) \cdots  \gp(x_k) \rangle_c =  {
 \sigma(\bar\rho)^{k-2}  \ \sigma^{(k)}(\bar{\rho})
 \over 2^{k-1} \,  L^{(k-1)d}}
\mathbf{G} (x_1 , \cdots  x_k) ,
\end{equation}
since at order $ \bJ^2$ the variations of   $\sigma(\bar\rho(x))$ with space can be neglected.
Defining $\mathbf{\Delta} = \sum_{i=1}^k \Delta_{x_i}$ the Laplacian on the domain $\bbD^k$ with Dirichlet boundary conditions, we deduce by the property of the Green function that 
\begin{equation}
\label{eq : G kpts-b property}
\mathbf{\Delta} \mathbf{G} (x_1 , \cdots , x_k) 
 =  -\bar{J}^2 (x_1)\ \delta_{x_1 = x_2, \cdots , x_k}.
\end{equation}

For $k=2$ and in  dimension $d=1$,  the steady state current   
$\bar{J}^2 (x) = \bJ^2$ is constant  and
the solution of (\ref{eq : G kpts-b property}) is
$$\mathbf{G}(x_1,x_2)={\bJ^2\over 2} \left[ x_1(1-x_2) 1_{x_1 \leq x_2} + x_2(1-x_1)  1_{x_1 > x_2} \right] $$  so that 
\begin{equation}
\label{eq : kpts-2}
 \langle \gp(x_1) \gp(x_2 )  \rangle_c =  { \sigma''(\bar{\rho})
 \over 4 \,  L}
\,   \bJ^2  \  {x_1 (1-x_2)}  \quad \text{for} \quad  x_1 < x_2 .
\end{equation}
This agrees with the exact computation for the SSEP 
\cite{Derrida 2007, Spohn} as $\sigma (\rho) = 2 \rho (1-\rho)$.

Notice that  when  $\sigma(\rho)$ is quadratic,  as in the SSEP or the KMP model, 
then all  long range correlations  in  \eqref{eq : G kpts-b property}  vanish, at order $\bJ^2 $, except pair correlations \eqref{eq : kpts-2}.
For a similar reason, if $\sigma(\rho)$ is a polynomial of degree $p$ (keeping $D(\rho)=1$), all the long range correlations  vanish if they  involve more  than $p$ points.

\medskip

In general, one cannot  get simpler expressions of the long range correlations 
\eqref{eq : kpts-b}.
Nevertheless using the identity \eqref{eq : G kpts-b property}, one can control the divergence of the cumulants when $x_1 , \cdots , x_k$ get closer. 
Consider for simplicity the one dimensional case $\bbD = [0,1]$ so that $\bf \Delta$ is a Laplacian on $[0,1]^k$ and the function $\bar{J}^2 (x) = \bJ^2$ is  constant. Then 
for $x_1 , \cdots , x_k$ away from the boundaries, the Green function $\mathbf{G}$
\eqref{eq : G kpts-b property} behaves as in $\bbR^k$ and we expect that for $k \geq 3$, it has a singularity of the form 
$$
\int_0^1 dx 
\frac{1}{ \big( \sum_{i =1}^k (x -x_i)^2 \big)^{(k-2)/2}} 
= 
\int_0^1 dx \frac{1}{ \big( k (x -m)^2 + k a\big)^{(k-2)/2}} 
$$
writing 
$$
m = \frac{1}{k} \sum_{i =1}^k x_i, \quad a = \frac{1}{k} \sum_{i =1}^k x_i^2 - m^2 \geq 0,
$$
where the parameter $a$ measures the distance between all the particles.
Thus when $a \to 0$, the singularity (see Figure \ref{fig2}) of the cumulants \eqref{eq : kpts-b}
 should diverge as 
\begin{equation}
\label{eq: asymptotic cumulants}
 \langle \gp(x_1) \gp(x_2 ) \cdots  \gp(x_k) \rangle_c   \ \propto \ 
{ \sigma^{(k)}(\bar{\rho})  \over L^{(k-1)}}
\times 
\begin{cases}
\log (a),  \quad & \text{if} \ k = 3\\
1/a^{(k-3)/2}, \quad & \text{if} \ k \geq  4
\end{cases}
\end{equation}
The sign of the correlation is determined by the macroscopic quantity $\sigma^{(k)}(\bar{\rho})$.
The prefactor depending on $L$ has been kept in \eqref{eq: asymptotic cumulants} to recall that the cumulant vanishes when $L \to \infty$ and that the parameter $a$ is constrained by the mesh size of the lattice so that the
asymptotics \eqref{eq: asymptotic cumulants} can be valid only for $a \gg 1/L^2$.

\begin{figure}[htbp]
\begin{center}
\includegraphics[width=4in]{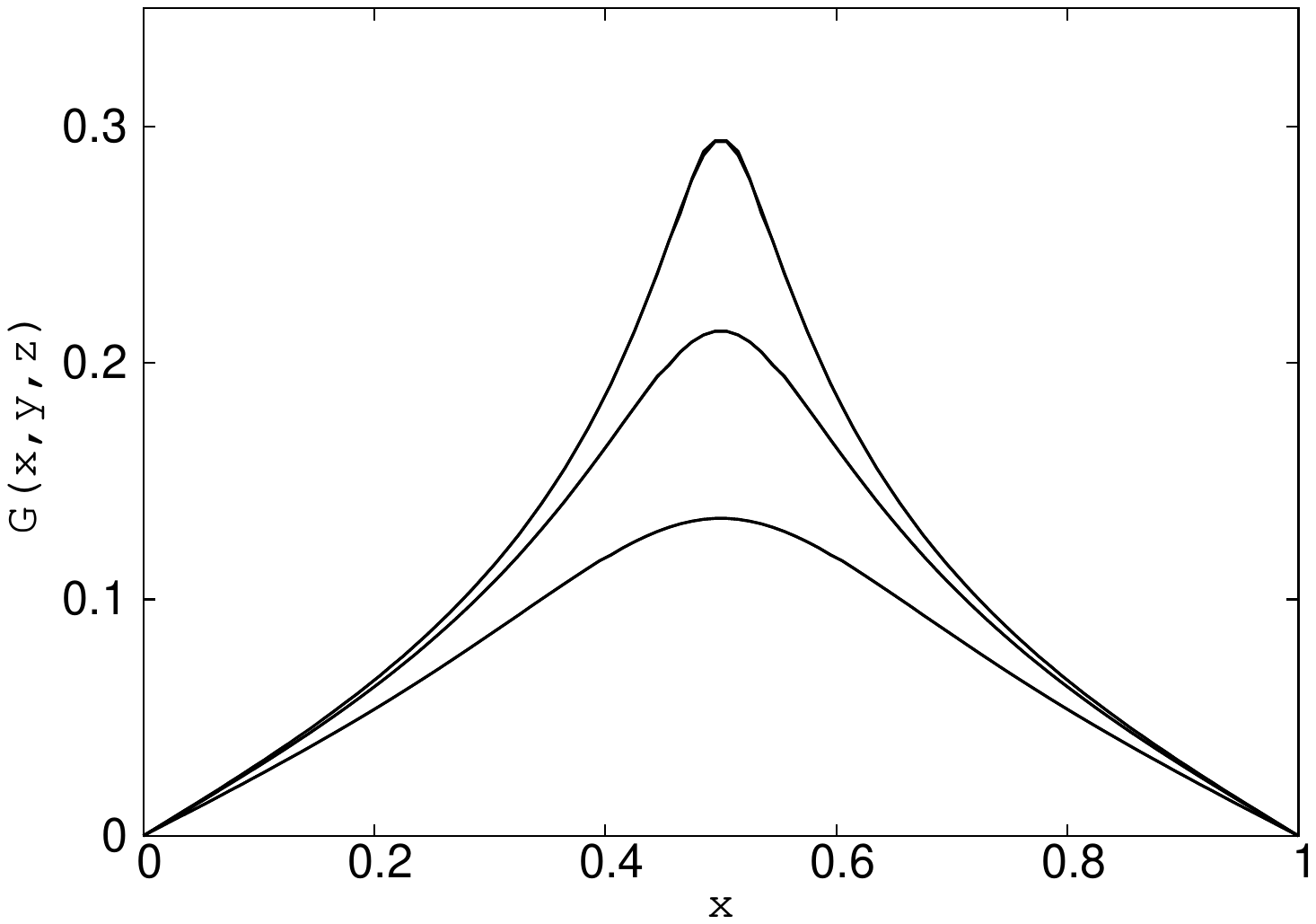} 
\caption{\small  In dimension 1 and for $k=3$ , the solution $G(x,y,z)$ of \eqref{eq : G kpts-b property} for $\bar{J}=1$ as a function of $x$ for $y=.4,z=.6$ (bottom curve), $y=.45,z=.55$ (middle curve) and $y=.475,z=.525$ (top curve). This indicates a singularity as the distance between the 3 points goes to zero, as expected in (\ref{eq: asymptotic cumulants}).
}
\label{fig2}
\end{center}
\end{figure}

\section{Conclusion}

For a diffusive system maintained out of equilibrium by reservoirs at the boundary of a domain $\bbD \subset \bbR^d$, we have considered here a regime close to equilibrium, i.e. a small drive from the reservoirs so that the flux flowing through the system in the stationary state is  of order $\bJ \ll 1$. 
Our main result \eqref{eq: LD ordre 3 d}, \eqref{eq: definition V bis d} is a perturbative analysis in $\bJ$ of the large deviation functional
for the stationary state.
So far the only expressions of the large deviation functional  had been computed for microscopic dynamics with specific transport coefficients. 
As stated in Claim \ref{Claim: HJ}, the large deviation functional for observing a macroscopic density profile $\gp$ in the stationary  state is of the form $\cS (\gp) = \cS^{local} (\gp)  + \cV(\gp) + O (\bJ^4)$ 
(see \eqref{eq: LD ordre 3 d}). The leading order $\cS^{local}$ has the same structure as the equilibrium large deviation function and in particular, it is local.
The functional $\cV$, defined in \eqref{eq: definition V bis d}, is non local : it encodes the non-equilibrium contributions at orders $\bJ^2$ and $\bJ^3$. 
The analysis relies on the macroscopic fluctuation theory, which  is expected to capture accurately features of diffusive microscopic dynamics at the macroscopic level. In particular, the long range correlations between $k$ sites at order 
$\bJ^2$ are predicted in \eqref{eq : kpts-b}. It is remarkable that these correlations diverge near the diagonal  as in \eqref{eq: asymptotic cumulants}.

Our perturbative analysis is limited to the first orders of the expansion in the flux $\bJ$. It would be interesting to try to extend it to higher orders. We expect that a similar perturbative approach can be also achieved for   diffusive dynamics driven away from equilibrium by a small varying external field in the bulk. 
Due to the non-local structure emerging in non-equilibrium systems, the large deviation functional is not  additive. For the SSEP, an additivity property  was  shown to hold in \cite{DLS2} and 
it is natural to ask whether  a similar structure could be derived for models with more general transport coefficients, at least in perturbative regimes.

\appendix

\section{Transport coefficients}
\label{Appendix : Transport coefficients}

In the fluctuating hydrodynamic equation \eqref{eq: fluctuating hydro}, the macroscopic transport coefficients $D,\sigma$   quantify the behavior of the macroscopic current. At the microscopic level, they can be interpreted as follows.
To fix a simple framework, consider a one dimensional lattice gas model on $\{ 1, \dots, L\}$ such that at each site $i$, the  number of particles is an integer $\eta_i$. These particles move according to a stochastic dynamics and the configuration at time $t$ is given by $(\eta_i(t), \  1 \leq i \leq L)$. In this work, we have in mind a system with a conservative bulk dynamics coupled to a creation-annihilation dynamics at the boundary acting as reservoirs imposing some prescribed densities.
On $\{1, \dots, L\}$, this means that the density at sites $1$ (resp. $L$) is $\rho_\text{left}$ (resp.  $\rho_\text{right}$). 
In particular, particles may enter into the system or be killed at sites $1, L$. 
Let $Q_t$ be the integrated current up to time $t$, i.e. the number of particles which 
go through the system during the time interval $[0,t]$ 
(note that given $L$,  the current can be measured anywhere in the bulk, and its asymptotic behavior  for $t$ large will be the same).
Suppose that the densities imposed by the reservoirs are 
$\rho_\text{left} = \rho + \Delta \rho $ and $\rho_\text{right}= \rho $ with 
$\Delta \rho$ small, then this  induces a flux of particles through the system and according to 
Fick's law, the current in the steady state is given at first order in $\Delta \rho$ by 
\begin{eqnarray}
\label{eq: diffusivite}
{\bbE [ Q_t ] \over t} = {1 \over L} D(\rho) \  \Delta \rho \, .
\end{eqnarray}
This defines the diffusion coefficient $D(\rho)$ at density $\rho$.
Suppose  that for $\rho_\text{left}= \rho_\text{right}=\rho$, then 
 on average
there is no flux through the system $\bbE [ Q_t ] =0$, 
nevertheless, the microscopic current fluctuates and  for large $t$, its variance is given asymptotically by 
\begin{eqnarray}
\label{eq: conductivite}
{\bbE[ Q_t^2 ] \over t} \simeq \frac{1}{L} \sigma(\rho) \, .
\end{eqnarray}
This defines the conductivity coefficient $\sigma(\rho)$ at density $\rho$.
We refer to \cite{Spohn book} for an in-depth discussion on these transport coefficients.
For  generic stochastic dynamics, there is no simple expression for the transport coefficients $D, \sigma$ as they are usually obtained implicitly in terms of a variational formula (see \cite[Chapter 7]{KL} and \cite{Arita Krapivsky Mallick}).
For some specific microscopic dynamics, known as gradient models \cite{KL},  the transport coefficients can be computed explicitly. This is for example the case of  the symmetric exclusion process 
($D(\rho)=1$ and $\sigma(\rho) = 2 \rho (1-\rho)$) or of the Zero-Range process ($D(\rho) = \sigma'(\rho)/2$).

\medskip

We conclude this appendix by deriving the Einstein relation \eqref{eq : Einstein relation} (see \cite[Part II, Eq. (2.72)]{Spohn book} or \cite[Section 11]{Derrida 2007} for alternative derivations).
Recall  that the large deviation functional $\cS$ should satisfy
the Hamilton Jacobi equation  (see  Claim \ref{prop: HJ} in  Appendix  
\ref{Appendix : Equations for the optimal profiles}
which was originally derived in \cite{BDGJL2}),
\begin{equation}
\label{eq: HJ1}
\int_\bbD   dx  \, \sigma (\gp) \Big(   \nabla \frac{\delta \cS}{\delta \gp} \Big)^2
- \int_\bbD dx   \,  (2 D( \gp ) \nabla \gp) \cdot  \nabla \frac{\delta \cS}{\delta \gp } = 0 .
\end{equation}
At equilibrium, the large deviation functional  is explicit
(\ref{eq: LD equilibrium}) and its functional derivative is given by 
\begin{equation}
\frac{\delta \cS}{\delta \gp}(\gp) = f'(\gp)-f'(\rho_*).
\label{eq: H equilibrium}
\end{equation}
Thus (\ref{eq: HJ1}) becomes for any profile $\gp$
\begin{equation}
\label{eq: HJ2}
\int_\bbD   dx  \, \sigma (\gp)  f''(\gp)^2  (\nabla \gp)^2
- \int_\bbD dx   \,  (2 D( \gp ) f''(\gp) (\nabla \gp)^2 = 0.
\end{equation}
This implies that the Einstein relation \eqref{eq : Einstein relation} should hold:
\begin{equation}
f''(\rho) \, \sigma(\rho)= 2 D(\rho).
\end{equation}

\section{Non-locality of $\cS$ and long range correlations}
\label{Appendix : Non-locality of S and long range correlations}

In this appendix, we  compute  the correlations of the density in the steady state from the knowledge   of  the functional $ \cS$ (\ref{eq: LD d}). To do so,  let us introduce $\cG$ the generating  function  of the density $\gp$ defined   in the large $L$ limit by
\begin{equation}
\label{eq: def Legendre G}
\cG(\{\lambda(x)\}) = \lim_{L \to \infty} { 1 \over L^d}  \log \left\langle \exp \left[L^d \int \lambda(x) \gp(x)  dx \right] \right\rangle .
\end{equation}
All the correlation functions can be obtained by taking successive derivatives of $\cG$ (assuming that these derivatives exist in the large $L$ limit)
\begin{align}
\langle \gp(x) \rangle = & \bar\rho(x)=  \left. {\delta \cG \over \delta \lambda(x)} \right|_{\lambda=0} ,
\nonumber \\
\langle \gp(x) \gp(y) \rangle_c = &  {1 \over L^d} \left. {\delta^2 \cG \over \delta  \lambda(x) \,  \delta \lambda(y) } \right|_{\lambda=0} , 
\label{eq: derivee cumulants}
\\
\langle \gp(x) \gp(y) \gp(z) \rangle_c = & {1 \over L^{2d}} \left. {\delta^3 \cG \over \delta  \lambda(x) \, \delta \lambda(y) \, \delta \lambda(z)} \right|_{\lambda=0} . \nonumber
\end{align}
The generating function $\cG$ is related to  the large deviation functional $\cS$ 
 by Varadhan’s Lemma
\cite[Chap 4.3]{Dembo_Zeitouni}
\begin{equation}
\label{eq: def Legendre G}
\cG( \lambda ) =\max_{ \{ \gp(x) \} }\left[ -\cS( \gp ) +  \int  \lambda(x) \gp(x)   dx \right] .
\end{equation}
Our goal is to recover the cumulants \eqref{eq: derivee cumulants} in terms of the functional $\cS$.
To study the derivates \eqref{eq: derivee cumulants} at $\lambda =0$,  it is enough to consider perturbations of $\cS$ around the steady state profile $\bar \rho$.
Let $r$ be the fluctuation  
$$r(x)=\gp(x)-\bar\rho(x)$$
and  assume that $\cS$ expands  in powers of $r$, near the steady state profile 
\begin{equation}
\label{eq : expansion S}
\cS (\gp) - \cS (\bar\rho)= {1 \over 2} \int_\bbD dx \int_\bbD dy \, \cS_2(x,y)  r(x) r(y) + {1 \over 6}  
 \int_\bbD dx \int_\bbD dy  \int_\bbD dz \,  \cS_3(x,y,z)  r(x) r(y) r(z)   + O(r^4)
\end{equation}
 (one can always choose  $S_2, S_3$ to be  symmetric functions).
Let $\cR_2$ be the inverse of the operator $\cS_2$
$$\int_\bbD dz \,  S_2(x,z)  \cR_2(z,y) = \delta(x-y).$$
Then for $\lambda$ small, the optimal value of $r$ in \eqref{eq: def Legendre G} satisfies \begin{equation}
\lambda(x) = \int_\bbD dy \, \cS_2(x,y)  r(y)+ \frac{1}{2} \int_\bbD dy  \int_\bbD dz \,  \cS_3(x,y,z) r(y) r(z)
+ O(r^3),
\end{equation}
so that at second order in $\lambda$
\begin{align*}
r (x) = &  \int_\bbD d u_1 \, \cR_2(x,u_1)  \lambda(u_1)\\
& - \frac{1}{2} \int_\bbD dy  \int_\bbD dz \int_\bbD du_1  \, \cR_2(x,u_1)  \cS_3(u_1,y,z) \int_\bbD du_2 \cR_2(y ,u_2)  \lambda(u_2) \int_\bbD du_3  \cR_2(z,u_3)  \lambda(u_3).
\end{align*}
From \eqref{eq: def Legendre G}, we know that $\frac{\partial \cG}{\partial \lambda (x)} = r(x) + \bar \rho(x)$  and we can deduce from the previous expression of $r$ an expansion of $\cG$ at the 3rd order in $\lambda$.
The cumulants are then obtained as the coefficients of this expansion
\begin{equation}
\label{eq : 2pts-a}
 \langle \gp(x) \gp(y) \rangle_c =  {1 \over L^d} \cR_2(x,y),
\end{equation}
\begin{equation}
 \langle \gp(x) \gp(y) \gp(z) \rangle_c = - {1 \over L^{2d}} \int_\bbD du_1 \int_\bbD d u_2 \int_\bbD du_3 \ \cR_2(x,u_1) \cR_2(y,u_2) R_2(z,u_3) \, \cS_3(u_1,u_2,u_3).
\label{eq : 3pts-a}
\end{equation}
More generally, denoting by $\cS_k$ the operator associated with the order $k \geq 3$ in \eqref{eq : expansion S}, then one can show that  when  $x_1  \dots  x_k$ are all different
\begin{equation}
 \langle \gp(x_1) \dots  \gp(x_k ) \rangle_c = - {1 \over L^{(k-1)d}} 
 \int_{\bbD^k} du_1 \dots  du_k  \, \cS_k(u_1, \dots ,u_k) \prod_{i =1}^k \cR_2(x_i,u_i) .
\label{eq : kpts-a 0}
\end{equation}

\medskip

In general, this leads to rather complicated formulas  for the correlations because one needs to invert the operator $\cS_2$. 
In the perturbative regime $\bJ \ll 1$, the main contribution of the functional $\cS$ \eqref{eq : S at order J2} has a local structure and the non local part occurs only at order  $\bJ^2$ 
 (see  \eqref{eq : S at order J2}) :
$$
\cS_2(x,y)=  f''(\bar\rho(x)) \delta(x-y) +  \cV_2(x,y),
$$
with $ \cV_2$ as in \eqref{eq : expansion of V 0}.
As $\cS_2$ is almost diagonal, one has at order $\bJ^2$
$$R_2(x,y)= {\delta(x-y) \over  f''(\bar\rho(x))}  - {\cV_2(x,y ) \over  f''(\bar\rho(x)) f''(\bar\rho(y))}, 
$$ 
so that (\ref{eq : 2pts-a}) becomes
\begin{equation}
\label{eq : 2pts}
 \langle \gp(x) \gp(y) \rangle_c = {1 \over L^d}  \left[ 
{\delta(x-y) \over  f''(\bar\rho(x))}  - {\cV_2(x,y ) \over  f''(\bar\rho(x)) f''(\bar\rho(y))} \right].
\end{equation}
Recalling notation \eqref{eq : S at order J2}, \eqref{eq : expansion of V 0}, the 3rd order reads 
$$\cS_3(x,y,z)=  f'''(\bar\rho(x)) \delta(x-y)\delta(x-z) +  \cV_3(x,y,z)$$
and at order $\bJ^2$, (\ref{eq : 3pts-a}) is given by
\begin{align}
 \langle \gp(x) \gp(y) \gp(z) \rangle_c = - {1 \over L^{2d}} & \left[  {f'''(\bar\rho(x)) \delta (x-y) \delta(x-z) \over f''(\bar\rho(x))^3}-
{f'''(\bar\rho(y)) \cV_2(x,y) \delta(y-z) \over f''(\bar\rho(x)) \, f''(\bar\rho(y))^3 } \right.
\nonumber
\\ 
& \ \ \
-{f'''(\bar\rho(z)) \cV_2(y,z) \delta(z-x) \over f''(\bar\rho(y)) \, f''(\bar\rho(z))^3 }
-{f'''(\bar\rho(x)) \cV_2(z,x) \delta(x-y) \over f''(\bar\rho(z)) \, f''(\bar\rho(x))^3 }
\nonumber \\ 
& \ \ \ \left. +  {\cV_3(x,y,z) \over  f''(\bar\rho(x)) f''(\bar\rho(y)) f''(\bar\rho(z))}
\right]. 
\label{eq : 3pts}
\end{align}
More generally by using formula \eqref{eq : kpts-a 0}, higher order cumulants can be obtained  from the knowledge of the expansion \eqref{eq : expansion of V 0} of $\cV$ 
\begin{equation}
\label{eq : kpts}
x_1 \not = \cdots \not = x_k, \qquad 
 \langle \gp(x_1) \gp(x_2 ) \cdots  \gp(x_k) \rangle_c = - {1 \over L^{(k-1)d}}  
{\cV_k(x_1 \dots x_k) \over  f''(\bar\rho(x_1)) \cdots f''(\bar\rho(x_k))}.
\end{equation}

\section{Equations for the optimal profiles}
\label{Appendix : Equations for the optimal profiles}

In this appendix, we  recall one way of obtaining the basic equations \eqref{eq: quasipotential}, \eqref{eq: evolution rho}, \eqref{eq: evolution H},
as well as the Hamilton-Jacobi equation 
(\ref{eq: HJ1}) which are standard in the macroscopic fluctuation theory.
\medskip

Let us first show that the optimal current has to be of the form (\ref{eq: optimal current}):
for any current $q(t,x)$ and density $\rho(t,x)$,  one can 
define a scalar function $H(x,t)$ as 
the unique solution (assuming $\gs(\rho)>0$) of
$$\nabla \cdot \big(\sigma(\rho)  \nabla H \big)= \nabla \cdot \big( q + D(\rho)  \nabla \rho \big)
\quad \quad \text{with} \quad \quad  H(t,y)=0  \ \ \ \text{ for}  \ \ \ y \in \partial \bbD  \ .$$ 
Then  one can write the current as 
$$q = - D \big( \rho  \big) \nabla \rho +  \sigma \big( \rho  \big) \nabla H + A$$
where 
 the vector $A(t,x)$   satisfies
$$\nabla \cdot A=0 .$$

Inserting this expression for $q$ into (\ref{eq: LD hydro1}), one gets
\begin{align*}
\cF_{[\tau_1,\tau_2]} (\rho) & =  \inf_q \int_{\tau_1}^{\tau_2} dt \int_\bbD dx 
\frac{\big(q (t,x) + D \big( \rho (t,x) \big) \nabla \rho (t,x) \big)^2}{2 \sigma \big( \rho (t,x) \big)} \\
& =  \inf_A \int_{\tau_1}^{\tau_2} dt \int_\bbD dx 
\frac{\big(\sigma \big( \rho \big) \nabla H + A \big)^2}{2 \sigma \big( \rho  \big)} \\
& = \inf_A \int_{\tau_1}^{\tau_2} dt \int_\bbD dx \left[
\frac{\sigma ( \rho )(\nabla H )^2}{2} 
+  \nabla H  \cdot A  + \frac{\big(A \big)^2}{2 \sigma \big( \rho \big)} \right] .
\end{align*}
It is easy to see that the cross term vanishes (using, after an integration by parts, that $H(t,x)$ vanishes at the boundary of $\bbD$  and the fact that  $\nabla \cdot A=0$). It is then clear that the minimum is  achieved for $A=0$. 
This justifies the choice of the decomposition (\ref{eq: optimal current}) and 
the form of the dynamical cost (\ref{eq: cost variational principle}).

\medskip

Given a density profile  $\rho_{\tau_1}$ at time $\tau_1$, we look for the optimal cost for observing the density $\gp$ at time $\tau_2$ 
\begin{equation}
\label{eq: cost variational principle v0}
\cF_{[\tau_1,\tau_2]} (\gp) = \inf_H \frac{1}{2} \int_{\tau_1}^{\tau_2} dt \int_\bbD dx \,  \sigma \big( \rho(t,x) \big) ( \nabla H (t,x) )^2 ,
\end{equation}
where $\rho$ follows an equation  (see (\ref{eq: conservation law},\ref{eq: optimal current}))  depending on the control function $H$
\begin{equation}
\label{eq: rho et H}
t \in [\tau_1,\tau_2], \qquad 
\partial_t \rho  = \nabla ( D(\rho) \nabla  \rho) - \nabla ( \sigma(\rho) \nabla  H),
\end{equation}
starting at $\rho (\tau_1)$ at time $\tau_1$ and at the boundary $\partial \bbD$, the density $\rho(t)$ is fixed equal to $\rho_\text{b}$ and $H(t)$ is equal to 0.
The control function $H$ is chosen such that $\rho (\tau_2) = \gp$.

\bigskip

\begin{Claim}
\label{Claim: trajectoire optimale}
Fix  $\rho(\tau_1)$  and $\rho(\tau_2) = \gp$ then  the optimal profile equations satisfy 
\begin{equation}
\label{eq: evolution equations bis}
t \in [\tau_1,\tau_2], \qquad 
\begin{cases}
\partial_t \rho  = \nabla \cdot ( D(\rho) \nabla  \rho) - \nabla \cdot  ( \sigma(\rho) \nabla  H) , \\
\partial_t H  = -  D(\rho) \Delta  H - \frac{1}{2} \sigma'(\rho) ( \nabla  H)^2,
\end{cases}
\end{equation}
with boundary terms $\rho(t,y) = \rho_\text{b}(y)$ and $H(t,y) =0$ for $y \in \partial \bbD$.
Furthermore
\begin{equation}
\label{eq: cF et H appendix}
H (\tau_2) = \frac{\delta \cF_{[\tau_1,\tau_2]}}{\delta \gp} (\gp) .
\end{equation}
\end{Claim}
\begin{proof}
Suppose that $H$ is an optimal control associated with the density $\rho$ by \eqref{eq: rho et H}, then we consider a perturbation of the form $H+h$ which leads to a perturbation of the trajectory $\rho + \psi$. Using \eqref{eq: rho et H}, this means that 
$\psi$ is given by 
\begin{equation}
\label{eq: psi et h}
\partial_t \psi  = \Delta ( D(\rho)  \psi) - \nabla \cdot ( \sigma'(\rho) \psi \nabla  H) -
 \nabla \cdot ( \sigma(\rho) \nabla  h).
\end{equation}
This perturbation modifies the dynamical large deviation function in \eqref{eq: cost variational principle v0} as follows
\begin{align}
\label{eq: cost variational principle h}
&\int_{\tau_1}^{\tau_2}  dt \int_\bbD dx \,   \Big(  \sigma ( \rho)  \nabla H  \cdot \nabla h
+ \frac{1}{2} \sigma'(\rho) \psi (\nabla  H)^2 \Big)  \nonumber\\
& \quad = \int_{\tau_1}^{\tau_2} dt  \int_\bbD dx \,   \Big( - H \,  \nabla \cdot ( \sigma ( \rho)    \nabla h )
+ \frac{1}{2} \sigma'(\rho) \psi (\nabla  H)^2   \Big) \nonumber  \\
& \quad
= \int_{\tau_1}^{\tau_2} dt \int_\bbD dx \,
  \Big(   \Big[ \partial_t \psi  - \Delta ( D(\rho)  \psi)
+ \nabla \cdot ( \sigma'(\rho) \psi \nabla  H) \Big]
+ \frac{1}{2} \sigma'(\rho) \psi (\nabla  H)^2 \Big)
   \nonumber \\
& \quad
= \int_{\tau_1}^{\tau_2} dt \int_\bbD dx \,
\psi
\Big( - \partial_t  H -  D(\rho)  \Delta H - \frac{1}{2} \sigma'(\rho)  (\nabla  H)^2 \Big)
 + \int_\bbD dx \,  H( \tau_2) \psi (\tau_2)  - \int_\bbD dx \,  H(\tau_1) \psi (\tau_1) ,
\end{align}
where we used \eqref{eq: psi et h} in the 3rd line and  integrations by parts in the last equality.

As the density $\rho$ is fixed at times $\tau_1$ and $\tau_2$,  one has 
$ \psi (\tau_1) = \psi (\tau_2) =0$. Since \eqref{eq: cost variational principle h} should be true for all $\psi$, we have recovered the optimal equations \eqref{eq: evolution equations bis}.

\medskip

We turn now to the proof of \eqref{eq: cF et H appendix} and consider a variation at the final time $\gp \to \gp + \delta \gp$.
Then \eqref{eq: cost variational principle h} gives at first order in $\delta \gp$
\begin{align*}
& \cF_{[\tau_1,\tau_2]} (\gp + \delta \gp) - \cF_{[\tau_1,\tau_2]} (\gp) \\
& = 
\int_{\tau_1}^{\tau_2} \int_\bbD dx \,    
\psi \Big( - \partial_s  H -  D(\rho)  \Delta H - \frac{1}{2} \sigma'(\rho)  (\nabla  H)^2 \Big)
+ \int_\bbD dx \,  H( \tau_2) \psi (\tau_2)  - \int_\bbD dx \,  H(\tau_1) \psi (\tau_1)  \\
& = 
\int_\bbD dx \,   H(\tau_2) \delta \gp ,
\end{align*}
where we used that $\psi (\tau_2) = \delta \gp$, $\psi (\tau_1) = 0$ and that $H$ is solution of \eqref{eq: evolution equations bis}.
By definition of the functional derivative, we get also at first order in $\delta \gp$
\begin{align*}
\cF_{[\tau_1,\tau_2]} (\gp + \delta \gp) - \cF_{[\tau_1,\tau_2]} (\gp) &= 
\int_\bbD dx \,  \frac{\delta \cF_{[\tau_1,\tau_2]}}{\delta \gp} (\gp) \, \delta \gp .
\end{align*}
Thus \eqref{eq: cF et H appendix} is proved by combining both identities  since $\delta \gp$ is arbitrary.
\end{proof}

The following result was  derived in \cite[Eq. (2.12)]{BDGJL2}.
\begin{Claim}
\label{prop: HJ}
For $\tau_1 = -\infty$ and $\rho_{-\infty} = \bar \rho$, then $\cS  (\gp) := \cF_{]-\infty,0]}  (\gp) $ is time independent and satisfies the following Hamilton-Jacobi equation :  for any density profile $\gp$
\begin{equation}
\label{eq: full HJ}
\int_\bbD   dx  \, \sigma (\gp) \Big(   \nabla \frac{\delta \cS}{\delta \gp} \Big)^2 
- \int_\bbD dx   \,  (2 D( \gp ) \nabla \gp)  \nabla \frac{\delta \cS}{\delta \gp } = 0.
\end{equation}
\end{Claim}

 To  establish  (\ref{eq: full HJ}), let $\rho(t,x)$ be  the  trajectory  which minimizes   (\ref{eq: quasipotential}), let $\gp-\delta \gp=\rho(-\delta t,x)$ where $\delta t$ and $\delta \gp$ are both small and let  $j(x)=q(0,x)$.
Then one has
$$\cS(\gp)  =  \inf_{\delta \gp} \left[ S(\gp-\delta \gp) +  \delta t \int_{\bbD} dx {(j + D(\gp) \nabla \gp)^2 \over 2 \sigma(\gp)} \right].
$$
 As  $\delta \gp= -  \delta t \,  \nabla \cdot j$
one has
$$0  =
 \delta t \inf_{\delta \gp}
 \int_{\bbD}  dx
 \left[ {\delta S \over \delta \gp} \nabla. j +
{(j + D(\gp) \nabla \gp)^2 \over 2 \sigma(\gp)} \right]
 =\delta t \inf_{\delta \gp}   \int_\bbD dx \left[ -j . \nabla {\delta S \over \delta \gp}  +    {(j + D(\gp) \nabla \gp)^2 \over 2 \sigma(\gp)} \right] ,
$$
where the last equality results from an integration by part and the fact  (see
(\ref{eq: cF et H appendix}))
that ${\delta S \over \delta \gp}=0$ on the boundary $\partial \bbD$. Then optimizing over $j$ (which is the same as optimizing over $\delta \gp$) leads to (\ref{eq: full HJ}).

\section{Check of Claim \ref{Claim: perturbative equations}}
\label{Appendix: Check of Claim}

We check now Claim \ref{Claim: perturbative equations}.
The gradient of the function $H$ defined in \eqref{eq: ansatz H d} is given by 
\begin{align}
\label{eq: perturbative H d gradient}
x \in \bbD, \qquad 
\nabla H (t,x) 
 =   \frac{2 D(\rho(t,x))}{\sigma(\rho(t,x))} \nabla \rho(t,x) 
+ \frac{2 \Jstat(x)}{\sigma(\bar \rho (x))}   
+ \nabla h(t,x),
\end{align}
where we used that the stationary current is $\Jstat(x) = -  D( \bar \rho(x)) \nabla \bar \rho(x)$.
Thus the evolution \eqref{eq: evolution equations d} of $\rho$ reads
\begin{align}
\partial_t \rho  & = - \nabla \cdot  ( D(\rho) \nabla  \rho) 
- 2  \nabla \cdot \left(  \frac{\sigma(\rho) }{\sigma(\bar \rho(x))}  \Jstat(x) \right)
- \nabla \cdot  ( \sigma(\rho) \nabla  h) \nonumber\\
& = - \nabla \cdot  ( D(\rho) \nabla  \rho) 
- 2   \Jstat(x)   \cdot \nabla  \left(  \frac{\sigma(\rho) }{\sigma(\bar \rho(x))} \right)
- \nabla \cdot  ( \sigma(\rho) \nabla  h),
\label{eq: rho evolution appendice}
\end{align}
as  the stationary current satisfies $\nabla \cdot \Jstat(x) =0$.
This shows the first equation in \eqref{eq: perturbative evolution equations 2}.

\medskip

We turn now to the derivation of the second equation in \eqref{eq: perturbative evolution equations 2} on the evolution of $h$.
Using \eqref{eq: rho evolution appendice}, the time derivative of $H$ is 
\begin{align}
\partial_t H & = \frac{2 D(\rho)}{\sigma(\rho)} \partial_t \rho + \partial_t h 
= 
\frac{2 D(\rho)}{\sigma(\rho)} \left( - \nabla ( D(\rho) \nabla  \rho) 
- 2 \Jstat(x) \cdot \nabla \left(  \frac{\sigma(\rho) }{\sigma(\bar \rho(x))}   \right)
- \nabla ( \sigma(\rho) \nabla  h) \right) + \partial_t h   \nonumber  \\
 & = - \frac{2 D(\rho)}{\sigma(\rho)}  \nabla ( D(\rho) \nabla  \rho)
- 4 \Jstat(x) \frac{D(\rho)}{\sigma(\rho)}  \frac{\sigma'(\rho) }{\sigma(\bar \rho)} \nabla  \rho
+ 4 \Jstat(x) D(\rho)   \frac{\sigma'(\bar \rho) }{\sigma(\bar \rho)^2} \nabla  \bar \rho 
\nonumber\\
 & \qquad \qquad - \frac{2 D(\rho)}{\sigma(\rho)}   \sigma'(\rho) \nabla \rho \nabla  h 
 - 2 D(\rho) \Delta  h  + \partial_t h.
 \label{eq: derivee H}
\end{align}
As $H$ is given by \eqref{eq: perturbative H d gradient} and that $\nabla \cdot \Jstat(x) =0$, we get  after tedious computations
\begin{align*}
-  D(\rho) &  \Delta  H -  \frac{1}{2} \sigma'(\rho) ( \nabla  H)^2\\
 & = -     \frac{2 D(\rho) }{\sigma(\rho)} \nabla \cdot  \big(  D(\rho)  \nabla \rho \big)
 -  D(\rho) \Delta h  -    \frac{2 D(\rho)\sigma'(\rho)}{\sigma(\rho)} \nabla \rho  \nabla h -   \frac{1}{2} \sigma'(\rho) \left( \nabla h \right)^2\\
 & 
 + 2  D(\rho) \frac{\sigma'(\bar \rho)}{\sigma(\bar \rho)^2}  \Jstat(x) \cdot \nabla  \bar \rho
-  4 \frac{\sigma'(\rho)}{\sigma(\bar \rho)}   \frac{ D(\rho)}{\sigma(\rho)} \Jstat(x)  \cdot \nabla \rho   
- 2   \frac{\sigma'(\rho)}{\sigma(\bar \rho)^2}   \big( \Jstat(x) \big)^2 
 -  2 \frac{\sigma'(\rho)}{\sigma(\bar \rho)}  \Jstat(x)  \cdot \nabla h .
\end{align*}
Combining this equation with \eqref{eq: derivee H},  we recover the second equation in \eqref{eq: perturbative evolution equations 2}
\begin{align*}
\partial_t h 
 & =    D(\rho) \Delta h 
 - \frac{1}{2} \sigma'(\rho) \left(   \nabla h \right)^2
- 2 \frac{\sigma'(\rho)}{\sigma(\bar \rho)}   \Jstat(x)  \cdot  \nabla h 
-   \frac{2 \big(\Jstat(x)  \big)^2  }{\sigma(\bar \rho)^2}    \left(  \sigma'(\rho)   
 -  \frac{D(\rho)}{ D(\bar \rho)}  \sigma'(\bar \rho)  \right) .
\end{align*}
This completes the Claim \ref{Claim: perturbative equations}.

\end{document}